\documentclass[journal,twoside]{IEEEtran}
\usepackage{cite}
\usepackage{amsmath,amssymb,amsfonts,mathtools,bm,mathrsfs}
\usepackage{graphicx}
\usepackage{subfig}
\usepackage{hyperref}
\hypersetup{hidelinks=true}
\usepackage{textcomp}
\usepackage{microtype}

\newtheorem{theorem}{Theorem}
\newtheorem{lemma}{Lemma}
\newtheorem{proposition}{Proposition}

\newtheorem{remark}{Remark}

\newcommand{\E}{\mathbb{E}}

\newcommand{\norm}[1]{\left\lVert #1 \right\rVert}
\newcommand{\abs}[1]{\left\lvert #1 \right\rvert}

\newcommand{\eps}{\varepsilon}

\newcommand{\Rcal}{\mathcal{R}}

\newcommand{\Scal}{\mathcal{S}}

\newcommand{\Kuu}{K^{uu}}
\newcommand{\Kuv}{K^{uv}}
\newcommand{\Kvu}{K^{vu}}
\newcommand{\Kvv}{K^{vv}}
\newcommand{\dd}{\mathrm{d}}

\newcommand{\cL}{\mathcal L}
\newcommand{\e}{\mathrm{e}}

\newenvironment{proof}{\begin{IEEEproof}}{\end{IEEEproof}}

\begin{document}
\title{Robust Stabilization of Linear Markov-Jumping Hyperbolic PDEs with Boundary Input Delay}
\author{Yihuai Zhang, Yidan Cao, Huan Yu, and Lu Liu
\thanks{This work was supported by the Research Grants Council of the Hong
Kong Special Administrative Region of China under Project CityU/11212225. (Corresponding author: Lu Liu.)}
\thanks{Yihuai Zhang, Yidan Cao, and Lu Liu are with the Department of Mechanical Engineering, City University of Hong Kong. (e-mail: yihuai.zhang@cityu.edu.hk, yidancao4-c@my.cityu.edu.hk, luliu45@cityu.edu.hk).}
\thanks{Huan Yu is with the Thrust of Intelligent Transportation, Hong Kong University of Science and Technology (Guangzhou). (e-mail: huanyu@hkust-gz.edu.cn).}
}

\maketitle
\begin{abstract}
This paper studies the robust stabilization of 2 $\times$ 2 linear hyperbolic partial differential equations (PDEs) with Markov-jumping parameters and boundary input delay. The main challenge arises from the simultaneous presence of stochastic parameter variations and input delay, which complicates  {both} the stability analysis and controller design. To address this issue, a nominal  {delay-compensating} backstepping controller is first designed for a fixed nominal system. Applying the nominal transformation to the stochastic system yields a target system with additional perturbation terms induced by parameter mismatch. A mode-independent Lyapunov functional is then constructed to establish  {a pathwise exponential estimate, which directly implies} mean-square exponential stability under an explicit small-mismatch condition. The proposed analysis provides a direct robustness certificate for nominal delay compensation without using mode-dependent Lyapunov functionals. Finally, we present simulation results  {and discuss how the conservative small-mismatch condition should be interpreted for the numerical example}.
\end{abstract}


\section{Introduction}
\label{sec:introduction}
\IEEEPARstart{H}{yperbolic} partial differential equations (PDEs) are widely  {used in} engineering  {applications} such as oil drilling~\cite{wang2020delay}, traffic flow~\cite{yu2022traffic},  {and} gas pipelines~\cite{bastin2016stability}. For this class of systems, boundary control is especially relevant because the control action often enters through the boundary. Over the past decades, the backstepping method has become one of the most effective tools for the stabilization of linear hyperbolic PDEs,  {because it transforms the original system into} an exponentially stable target system  {that can be analyzed using} Lyapunov  {methods}~\cite{krstic2008boundary}.

In many applications, however, the model coefficients are not deterministic. Transport speeds, source terms, and boundary couplings may vary because of uncertain operating conditions, disturbances, or abrupt mode changes. A convenient and mathematically tractable representation of such uncertainty is to model the time-varying coefficients as finite-state Markov-jumping processes. This leads to stochastic hyperbolic systems whose analysis must account not only for the distributed dynamics but also for random switching among different modes. Stochastic linear hyperbolic PDEs  {have} been widely investigated~\cite{amin2011exponential,bolzern2006almost,do2012continuous,lamare2015switching,prieur2014stability}. 
In~\cite{wang2012stochastically}, the authors investigated the robust stochastically exponential stability and stabilization of uncertain linear hyperbolic PDEs with Markov-jumping parameters by using linear matrix inequalities~(LMIs). Prieur~\cite{prieur2014stability} analyzed changes  {in the} boundary conditions and derived sufficient conditions for the exponential stability of the switching system. LMIs were then applied to obtain the sufficient conditions for stochastic stabilization of traffic flow described by hyperbolic PDEs~\cite{zhang2017stochastic}. Furthermore, Zhao~\cite{zhao2020boundary} investigated the output feedback stabilization of PDE-ODE cascade systems with stochastic jumps. By adopting the backstepping method, Auriol~\cite{auriol2023mean} demonstrated the mean-square exponential stability of coupled hyperbolic systems with Markov-jumping parameters, where a mode-dependent Lyapunov functional  {was} constructed. Then  {this approach} was extended to the mixed-autonomy traffic system in~\cite{zhang2023mean},  {and to} the use of neural operators~\cite{ZhangOperatorLearning2026}.

In addition, another feature that frequently appears in practice is the actuation delay. Boundary control commands may be delayed by computation, communication, sensing, or actuator dynamics. A standard and effective way to handle this issue is to rewrite the delay as a transport equation and to analyze the resulting augmented system~\cite{krstic2008backstepping}. This approach eliminates the explicit delayed argument from the boundary condition and places the delay entirely into an auxiliary transport state, which can then be treated together with the plant in a cascade framework. Auriol~\cite{auriol2018delay} studied the delay-robust control of hyperbolic PDEs  {and} derived the trade-off between the boundary reflection margin and finite-time convergence. The authors in~\cite{kong2022prediction} considered the stochastic input delay of a linear time-invariant system and derived sufficient conditions for mean-square exponential stability. For the case of hyperbolic PDEs, Zhang~\cite{ZhangRobustStabilization2023} investigated the robust stabilization of linear hyperbolic PDEs with uncertain boundary input delay using the backstepping method.  {Stochastic} input delay for parabolic PDEs was investigated in~\cite{guan2024robustness}. Despite the significant progress in the stabilization of hyperbolic PDEs with either Markov-jumping parameters or input delay, the problem of robust stabilization of linear hyperbolic PDEs with both Markov-jumping parameters and boundary input delay has not been addressed. The simultaneous presence of stochastic parameter variations and input delay introduces new challenges in the stability analysis and controller design. This motivates the current work,  {which develops} a backstepping-based approach for the robust stabilization of linear hyperbolic PDEs with both Markov-jumping parameters and boundary input delay.

\textbf{Contributions}. The contribution of this paper is twofold. First, we derive the target system obtained by applying a nominal  {delay-compensating} backstepping transformation to a Markov-jumping hyperbolic system with boundary input delay  {and prove the well-posedness of the closed-loop solution generated by the delayed boundary feedback}. Second, we construct a common Lyapunov functional for the augmented target system and obtain a mode-independent small-mismatch condition for  {pathwise, and hence} mean-square exponential stability. Unlike mode-dependent Lyapunov approaches, the proposed analysis avoids Markov jump terms in the generator, at the cost of a conservative uniform mismatch bound.  {It gives a robust common-Lyapunov certificate for Markov-jumping hyperbolic PDEs under a nominal delay-compensating controller.}

The paper is organized as follows. Section~\ref{sec:problemstatement} formulates the problem, including the stochastic system and the nominal  {delay-compensating} controller design. Section~\ref{sec:StochasticSystem} derives the stochastic target system under the nominal backstepping transformation. Section~\ref{sec:LyapunovAnalysis} constructs a mode-independent Lyapunov functional and establishes mean-square exponential stability of the closed-loop system. Section~\ref{sec:simulation} gives simulation results to illustrate the effectiveness of the proposed approach  {and reports the mismatch level used in the example}. Finally, Section~\ref{sec:conclusion} concludes the paper and outlines future research directions.

\section{Problem statement}
\label{sec:problemstatement}
We consider a  $2 \times 2$  linear Markov-jumping hyperbolic system with input delay
\begin{align}
 \partial_t u(x, t)+\lambda(t) \partial_x u(x, t)&=\sigma^{+}(t) v(x, t),\label{eq:stosys1}\\
 \partial_t v(x, t)-\mu(t) \partial_x v(x, t)&=\sigma^{-}(t) u(x, t),
\end{align}
with boundary conditions 
\begin{align}
 u(0,t)&=\phi(t) v(0, t), \\
 v(1,t)&=U(t-D),\label{eq:stosys4}
\end{align}
where the spatial and time variables $(x,t)$ belong to $\{[0,1]\times \mathbb{R}^+ \}$. $D>0$ denotes a constant and known input delay. The input history is initialized as $U(h) = 0$ for $h \in [-D,0]$.
The stochastic characteristic speeds $\lambda(t)>0$ and $\mu(t)>0$ are time-varying. The in-domain couplings $\sigma^{+}(t)$, $\sigma^{-}(t)$ and the boundary coupling $\phi(t)$ are also stochastic. The stochastic coefficients are modeled through a single joint finite-state Markov process. Let $\theta(t)$ be a right-continuous Markov process with finite state space $\Rcal:=\{1,\ldots,r\}$.  For \(0\le s\le t\), define the transition probabilities
$P_{ij}(s,t):=\mathbb P\{\theta(t)=j\mid \theta(s)=i\}, i,j\in\Rcal.$
They satisfy the forward Kolmogorov equation~\cite{kolmanovsky2001mean,hoyland2009system,ross2014introduction}
\begin{align}
    \frac{\partial}{\partial t}P_{ij}(s,t) =
\sum_{k\ne j}P_{ik}(s,t)\tau_{kj}(t) - P_{ij}(s,t)\sum_{\ell\ne j}\tau_{j\ell}(t),
\end{align}
with \(P_{ii}(s,s)=1\) and \(P_{ij}(s,s)=0\) for \(i\ne j\).
Here \(\tau_{ij}(t)\ge 0\), \(i\ne j\), denotes the transition rate from
mode \(i\) to mode \(j\), and \(\tau_{ii}(t)=0\). We assume that the
transition rates are bounded, namely there exists
\(\tau^\star>0\) such that $\sum_{\ell\ne i}\tau_{i\ell}(t)\le \tau^\star, \forall i\in\Rcal, \forall t\ge 0$.
Hence the Markov process is non-explosive.

The Markov-jump parameter vector is defined by
\begin{align}
    \delta(t):=\delta_{\theta(t)},
\delta_j=(\lambda_j,\mu_j,\sigma_j^+,\sigma_j^-,\phi_j), j\in\Rcal.
\end{align}
Equivalently, when \(\theta(t)=j\), the active coefficients in
\eqref{eq:stosys1}-\eqref{eq:stosys4} are $\lambda(t)=\lambda_j,\mu(t)=\mu_j, \sigma^+(t)=\sigma_j^+, \sigma^-(t)=\sigma_j^-,
\phi(t)=\phi_j$.
We assume that the characteristic speeds are uniformly positive and bounded, i.e., $0<\underline{\lambda}\le \lambda_j\le \bar{\lambda}, 0<\underline{\mu}\le \mu_j\le \bar{\mu}, j\in\Rcal$, and the coupling coefficients and boundary coefficient are also uniformly bounded $\underline{\phi}\leq \phi_j \leq \bar{\phi}$, $\underline{\sigma}^\pm \leq\sigma_j^\pm \leq \bar{\sigma}^\pm, j\in\Rcal$.

For the nominal system used in the controller design, fix
\begin{align}
    \delta_0=(\lambda_0,\mu_0,\sigma_0^+,\sigma_0^-,\phi_0),
\end{align}
where \(\lambda_0>0\), \(\mu_0>0\), and \(\phi_0\ne 0\).
The nominal vector \(\delta_0\) may be chosen as one of the Markov modes, but this is not required. Let
\begin{align}
    \mathcal S:=\{\lambda,\mu,\sigma^+,\sigma^-,\phi\}.
\end{align}
For each \(X\in\mathcal S\), denote by \(X_j\) the \(X\)-component of
the joint mode \(\delta_j\), and by \(X_0\) the corresponding component
of \(\delta_0\). The mismatch between mode \(j\) and the nominal plant is
defined as
\begin{align}\label{eq:modemismatch}
    \Delta_j:=\sum_{X\in\Scal}\abs{X_j-X_0}, j \in \Rcal.
\end{align}
We also define the worst-case mismatch
\begin{align}\label{eq:Delta_max}
    \Delta_{\max}:=\max_{j\in\Rcal}\Delta_j.
\end{align}
This joint-chain formulation includes independent componentwise Markov chains as a special case, by taking \(\Rcal\) to be the Cartesian product of the component state spaces. In the present work, however, no independence among the parameter components is required.

We first design the  {delay-compensating} controller for the fixed nominal system associated with \(\delta_0\). The nominal system with input delay is written as
\begin{align}
    \partial_t u_{\text{nom}}(x,t)+\lambda_0\partial_x u_{\text{nom}}(x,t)&=\sigma_0^+v_{\text{nom}}(x,t),\label{eq:nom_u}\\
    \partial_t v_{\text{nom}}(x,t)-\mu_0\partial_x v_{\text{nom}}(x,t)&=\sigma_0^-u_{\text{nom}}(x,t),\label{eq:nom_v}
\end{align}
with boundary conditions
\begin{align}
    u_{\text{nom}}(0,t)&=\phi_0v_{\text{nom}}(0,t),\label{eq:nom_left_bc}\\
    v_{\text{nom}}(1,t)&=U_{\text{nom}}(t-D).\label{eq:nom_right_bc}
\end{align}

By using a transport PDE to express the input delay, the system~\eqref{eq:nom_u}-\eqref{eq:nom_right_bc} is rewritten as follows
\begin{align}
    \partial_t u_{\text{nom}}(x,t)+\lambda_0\partial_x u_{\text{nom}}(x,t)&=\sigma_0^+v_{\text{nom}}(x,t),\label{eq:tsnom_u}\\
    \partial_t v_{\text{nom}}(x,t)-\mu_0\partial_x v_{\text{nom}}(x,t)&=\sigma_0^-u_{\text{nom}}(x,t),\label{eq:tsnom_v}\\
    u_{\text{nom}}(0,t)&=\phi_0v_{\text{nom}}(0,t),\label{eq:tsnom_left_bc}\\
    v_{\text{nom}}(1,t)&=\eta(0,t),\label{eq:tsnom_right_bc}\\
    D \partial_t \eta(x,t) &= \partial_x \eta(x,t),\label{eq:transport}\\
    \eta(1,t)&= U_{\text{nom}}(t). \label{eq:ts_right_bc}
\end{align}
Following the backstepping design, two Volterra transformations for the system states and the delay-channel state are applied to the delayed-nominal system~\eqref{eq:tsnom_u}-\eqref{eq:ts_right_bc}, we have
\begin{align}
    &\gamma_{\text{nom}}(x,t) = w_{\text{nom}}(x,t) - \int_0^x \mathbf{K}(x,\xi) w_{\text{nom}}(\xi,t) d\xi, \label{eq:plant_BS}\\
    &z(x,t) = \eta(x,t) - \int_0^x p(x-\xi) \eta(\xi,t) d\xi \nonumber\\
    &-\int_0^1 q_1(x,\xi) u_{\text{nom}}(\xi,t) d\xi - \int_0^1 q_2(x,\xi) v_{\text{nom}}(\xi,t) d\xi\label{eq:delay_BS},
\end{align}
where $\gamma(x,t) = (\alpha_{\text{nom}}(x,t), \beta_{\text{nom}}(x,t))^\top$, $w(x,t) = (u_{\text{nom}}(x,t), v_{\text{nom}}(x,t))^\top$. Here, the kernels $K^{\cdot\cdot}$ are defined on the triangular domain $\mathcal{T}_1 = \{(x,\xi): 0 \leq \xi \leq x \leq 1\}$ and satisfy the kernel equations which can be found in ~\cite{vazquez2011backstepping} and in its compact form as $\mathbf{K}(x,\xi)= \begin{bmatrix}
\Kuu(x,\xi) & \Kuv(x,\xi)\\
\Kvu(x,\xi) & \Kvv(x,\xi)
\end{bmatrix}$. $q_1$, and $q_2$ are kernels defined on the rectangular domain $\mathcal{T}_2 = \{0 \leq x, \xi \leq 1\}$. The kernel function $p(x)$ is defined on $x\in [0,1]$. Using the backstepping transformation, the target system is given by
\begin{align}
    \partial_t \alpha_{\text{nom}}(x,t)+\lambda_0\partial_x \alpha_{\text{nom}}(x,t)&=0, \label{eq:nomtarget1}\\
    \partial_t \beta_{\text{nom}}(x,t)-\mu_0\partial_x \beta_{\text{nom}}(x,t)&=0,\\
    \alpha_{\text{nom}}(0,t)&=\phi_0\beta_{\text{nom}}(0,t),\\
    \beta_{\text{nom}}(1,t)&=z(0,t),\\
    D \partial_t z(x,t) &= \partial_x z(x,t),\\
    z(1,t) &= 0.\label{eq:nomtarget6}
\end{align}
The nominal  {delay-compensating} controller is then given by
\begin{align}\label{eq:nomcontrol}
    U_{\text{nom}}(t)=&\int_0^1 p(1-\xi)\eta(\xi,t)\, d\xi
  +\int_0^1 q_1(1,\xi)u_{\text{nom}}(\xi,t)\, d\xi \nonumber\\
  &+\int_0^1 q_2(1,\xi)v_{\text{nom}}(\xi,t)\, d\xi.
\end{align}
The kernels $p$, $q_1$, and $q_2$ are determined by the backstepping design. The kernels $q_1$ and $q_2$ satisfy the following kernel equations
\begin{align}
    &\tfrac{1}{D} \partial_x q_1(x,\xi) - \lambda_0 \partial_\xi q_1(x,\xi) = \sigma_0^- q_2(x,\xi), \label{eq:q1_kernel}\\
    &\tfrac{1}{D} \partial_x q_2(x,\xi) + \mu_0 \partial_\xi q_2(x,\xi) = \sigma_0^+ q_1(x,\xi),\label{eq:q2_kernel}
\end{align}
with boundary conditions
\begin{align}
    q_1(x,1) &= 0,\\
    \frac{\mu_0}{\phi_0\lambda_0}q_2(x,0) &= q_1(x,0),\\
    q_1(0,\xi) &= \Kvu(1,\xi),\\
    q_2(0,\xi) &= \Kvv(1,\xi).
\end{align}
The kernel function $p(x)$ is obtained as
\begin{align}
    p(x)= D\mu_0 q_2(x,1).
\end{align}
Building on the results in~\cite{ZhangRobustStabilization2023}, the following proposition for the bound of the kernels can be obtained.
\begin{proposition}\label{thm:kernelbound}
The backstepping kernels $K^{\cdot\cdot}$, $q_1$, $q_2$, and $p$
appearing in the transformations \eqref{eq:plant_BS}--\eqref{eq:delay_BS} are well-defined and bounded.
Moreover, the first-order derivatives of the kernels that appear in
the perturbation coefficients are also bounded. More precisely, there
exist positive constants $\bar K$, $\bar K_1$, $\bar q$, $\bar q_\xi$,
and $\bar p$, depending only on the nominal parameters and on the
delay $D$, such that
\begin{align}
|K^{\cdot\cdot}(x,\xi)| &\le \bar K, \nonumber\\
|\partial_x K^{\cdot\cdot}(x,\xi)|+
|\partial_\xi K^{\cdot\cdot}(x,\xi)| &\le \bar K_1, (x,\xi)\in\mathcal T_1, \nonumber\\
|q_i(x,\xi)| &\le \bar q, i=1,2, \nonumber\\
|\partial_\xi q_i(x,\xi)| &\le \bar q_\xi, i=1,2,
(x,\xi)\in\mathcal T_2, \nonumber\\
|p(x)| &\le \bar p, x\in[0,1].
\end{align}
\end{proposition}

\begin{proof}
The boundedness of backstepping kernels
$K^{\cdot\cdot}$ and of their first-order derivatives follows from the standard well-posedness and regularity results for Volterra kernel equations associated with $2 \times 2$ hyperbolic systems~\cite{ZhangRobustStabilization2023,CoronLocal2013}. Similarly, the kernels $q_1$ and $q_2$ are solutions of a
linear hyperbolic kernel system on the rectangle domain
$\mathcal T_2$, with boundary conditions determined by the bounded kernels $K^{vu}$ and $K^{vv}$. Hence $q_1$, $q_2$, and the derivatives $\partial_\xi q_1$, $\partial_\xi q_2$ are uniformly bounded on $\mathcal T_2$. Finally, since $p(x)=D\mu_0 q_2(x,1)$,
the boundedness of $p$ follows immediately from the boundedness of $q_2$. This completes the proof of Proposition~\ref{thm:kernelbound}.
\end{proof}

\begin{proposition}\label{thm:nominal_finite_time}
Assume the kernels of Proposition~\ref{thm:kernelbound} are used in the controller \eqref{eq:nomcontrol}. The nominal system~\eqref{eq:nom_u}-\eqref{eq:nom_right_bc} with initial conditions $u_{\text{nom}}^0$, $v_{\text{nom}}^0$, $\eta_{\text{nom}}^0 \in L^2([0,1])^3$ is finite-time stable. More precisely, letting
\begin{equation}\label{eq:Tstar_nominal}
    t_f:=D+\frac{1}{\mu_0}+\frac{1}{\lambda_0},
\end{equation}
then the system reaches the zero state in finite time $t_f$.
\end{proposition}
\begin{proof}
In the scaled formulation the actuator state satisfies $z_t-\frac{1}{D}z_x=0$ with boundary condition $z(1,t)=0$. The associated characteristic curves satisfy $x(t)=1-(t-t_0)/D$, so the characteristic formula gives $z(x,t)=0$ for all $t\ge D(1-s)$. In particular, $z(0,t)=0$ for all $t\ge D$. For such times the boundary condition for $\beta$ becomes $\beta(1,t)=0$. Since $\beta_t-\mu_0\beta_x=0$, characteristic propagation implies $\beta(\cdot,t)\equiv0$ for all $t\ge D+1/\mu_0$. The left boundary condition of $\alpha$ is then $\alpha(0,t)=\phi_0\beta(0,t)=0$, and $\alpha_t+\lambda_0\alpha_x=0$ gives $\alpha(\cdot,t)\equiv0$ for all $t\ge D+1/\mu_0+1/\lambda_0$. Hence \eqref{eq:Tstar_nominal} holds. This finishes the proof of Proposition~\ref{thm:nominal_finite_time}.
\end{proof}

Now we first state the well-posedness of the closed-loop system with Markov-jumping parameters and input delay under the nominal  {delay-compensating} controller, we have the following lemma.
\begin{lemma}\label{lem:wellposedness}
    For any initial conditions of the stochastic system~\eqref{eq:stosys1}-\eqref{eq:stosys4} with $u^0$, $v^0, \eta^0 \in L^2([0,1])^3$, the closed-loop system with the nominal  {delay-compensating} controller~\eqref{eq:nomcontrol} is well-posed in the sense that there exists a unique solution $(u,v,\eta)$ such that for any initial mode $\delta(0)$ and $t \geq 0$, 
    \begin{align}~\label{eq:wellposedness}
        \E_{[0,(u^0, v^0,\eta^0,\delta(0))]}[\norm{u(\cdot,t),v(\cdot,t),\eta(\cdot,t)}_{L^2}^2] < \infty.
    \end{align}
\end{lemma}
\begin{proof}
Let \(\mathcal H:=L^2([0,1])^3\) and write $W(t)=(u(\cdot,t),v(\cdot,t),\eta(\cdot,t))^\top\in\mathcal H$.
The nominal  {delay-compensating} controller is implemented with the input $U(t)=\mathcal K W(t)$, where $\mathcal K W = \int_0^1 p(1-\xi)\eta(\xi)\,\dd \xi + \int_0^1 q_1(1,\xi)u(\xi)\,\dd \xi + \int_0^1 q_2(1,\xi)v(\xi)\,\dd \xi$.
Since \(p,q_1,q_2\) are bounded, \(\mathcal K:\mathcal H\to\mathbb R\) is a bounded linear functional. For each fixed mode \(j\in\mathcal R\), freeze the Markov mode at
\(\theta(t)=j\). The corresponding closed-loop operator
\(\mathcal A_j\) is given by 
\begin{align}
    \mathcal A_j
\begin{pmatrix}
u\\ v\\ \eta
\end{pmatrix} = \begin{pmatrix}
-\lambda_j u_x+\sigma_j^+ v\\
\mu_j v_x+\sigma_j^- u\\
D^{-1}\eta_x
\end{pmatrix},
\end{align}
with domain $D(\mathcal A_j)=
\{(u,v,\eta)\in \mathcal{H}:
u(0)=\phi_j v(0),\
v(1)=\eta(0),\
\eta(1)=\mathcal K W\}$.
This is a linear hyperbolic system with bounded linear
boundary feedback. Therefore, each
\(\mathcal A_j\) generates a \(C_0\)-semigroup
\(\mathcal M_j(t)\) on \(\mathcal H\). Since the mode set
\(\mathcal R\) is finite, there exist constants \(M_s\ge 1\) and
\(\omega_s\in\mathbb R\), independent of \(j\), such that
\begin{align}
    \|\mathcal M_j(t)\|_{\mathcal L(\mathcal H)}
\le M_s e^{\omega_s t}, t\ge 0, j\in\mathcal R .
\end{align}
Let \(0=T_0<T_1<T_2<\cdots\) be the jump times of \(\theta(t)\).
Because the total transition rate is uniformly bounded by
\(\tau^\star\), the Markov chain is non-explosive, in particular, for
every \(T>0\), the number of jumps
\begin{align}
    N_T:=\#\{k\geq 1: T_k\le T\}
\end{align}
is finite and is dominated by a Poisson random variable with parameter \(\tau^\star T\). For a fixed sample path and for \(t\in[T_k,T_{k+1})\), define the
solution recursively by
\begin{align}
    W(t)=\mathcal M_{\theta(T_k)}(t-T_k)W(T_k),
\end{align}
with the assumption that the state is not reset at jump times, i.e., \(W(T_k)=W(T_k^-)\). Since \(N_T<\infty\) almost surely on each compact interval \([0,T]\), this concatenation gives a unique pathwise solution on \([0,T]\). Uniqueness follows from the uniqueness of the semigroup solution on each inter-jump interval and the continuity of the state at the jump times. Moreover, for any \(t\in[0,T]\),
\begin{align}
    \|W(t)\|_{\mathcal H} \le M_s^{N_T+1}e^{\omega_s T}\|W(0)\|_{\mathcal H}.
\end{align}
Hence
\begin{align}
    \mathbb E\|W(t)\|_{\mathcal H}^2
\le M_s^2e^{2\omega_s T}
\mathbb E\!\left[M_s^{2N_T}\right] \|W(0)\|_{\mathcal H}^2
<\infty.
\end{align}
Since \(T>0\) is arbitrary, the closed-loop system is well posed for all \(t\ge0\) and satisfies~\eqref{eq:wellposedness}. This completes the proof of Lemma~\ref{lem:wellposedness}.
\end{proof}
\section{Stochastic target system}
\label{sec:StochasticSystem}
We now consider the original stochastic system with the nominal backstepping transformation. For notational simplicity, we write
\begin{align}
    Y(x,t)=(\alpha(x,t),\beta(x,t),z(x,t)).
\end{align}
Under the nominal backstepping transformation, the stochastic target system at the joint mode $\delta(t) = \delta_j$ is given by
\begin{align}
    \partial_t\alpha(x,t)+\lambda_j\partial_x\alpha(x,t)&=\mathcal F_j(x,t), \label{eq:tarsto1}\\
  \partial_t\beta(x,t)-\mu_j\partial_x\beta(x,t)&=\mathcal G_j(x,t), \label{eq:tarsto2}\\
  \alpha(0,t)&=\phi_j\beta(0,t), \label{eq:tar_bc1}\\
  \beta(1,t)&=z(0,t), \label{eq:tar_bc2}\\
  D\partial_t z(x,t)&=\partial_x z(x,t)+\mathcal H_j(x,t), \label{eq:tarsto3}\\
  z(1,t)&=0, \label{eq:tar_bc3}
\end{align}
where 
\begin{align}
  \mathcal F_j(x,t)
  &= f_{1j}\,v(x,t) + f_{2j}(x)\,\beta(0,t) \notag\\
  &\quad + \int_0^x \bigl[ f_{3j}(x,\xi)\,u(\xi,t) + f_{4j}(x,\xi)\,v(\xi,t) \bigr]\,\dd \xi, \label{eq:Fj}\\
  \mathcal G_j(x,t)=& g_{1j}u(x,t)+g_{2j}(x)\beta(0,t) \nonumber\\
  &+\int_0^x g_{3j}(x,\xi)u(\xi,t)+g_{4j}(x,\xi)v(\xi,t)\,\dd\xi, \label{eq:Gj}\\
  \mathcal H_j(x,t)=& h_{1j}(x)\eta(0,t)+h_{2j}(x)\beta(0,t) \nonumber\\
  &+\int_0^1 h_{3j}(x,\xi)u(\xi,t)+h_{4j}(x,\xi)v(\xi,t)\,\dd\xi. \label{eq:Hj}
\end{align}
The perturbation coefficients are
\begin{align}
 f_{1j}&=\sigma^+_j-\sigma^+_0\frac{\lambda_j+\mu_j}{\lambda_0+\mu_0}, \label{eq:f1}\\
 f_{2j}(x)&=\left(\mu_j-\frac{\lambda_j\phi_j\mu_0}{\lambda_0\phi_0}\right)K^{uv}(x,0), \label{eq:f2}\\
 f_{3j}(x,\xi)&=\left(\frac{\lambda_j}{\lambda_0}\sigma^-_0-\sigma^-_j\right)K^{uv}(x,\xi), \label{eq:f3}\\
 f_{4j}(x,\xi)&=(\lambda_0-\lambda_j)\partial_xK^{uv}(x,\xi)
 +(\mu_j-\mu_0)\partial_\xi K^{uv}(x,\xi) \nonumber\\
 &\quad -(\sigma^+_j-\sigma^+_0)K^{uu}(x,\xi), \label{eq:f4}
 \end{align}
\begin{align}
 g_{1j}&=\sigma^-_j-\frac{\lambda_j+\mu_j}{\lambda_0+\mu_0}\sigma^-_0, \label{eq:g1}\\
 g_{2j}(x)&=\left(-\lambda_j\phi_j+\mu_j\frac{\lambda_0\phi_0}{\mu_0}\right)K^{vu}(x,0), \label{eq:g2}\\
 g_{3j}(x,\xi)&=(\mu_j-\mu_0)\partial_xK^{vu}(x,\xi)
 -(\lambda_j-\lambda_0)\partial_\xi K^{vu}(x,\xi) \nonumber\\
 &\quad -(\sigma^-_j-\sigma^-_0)K^{vv}(x,\xi), \label{eq:g3}\\
 g_{4j}(x,\xi)&=\left(\frac{\sigma^+_0\mu_j}{\mu_0}-\sigma^+_j\right)K^{vu}(x,\xi), \label{eq:g4}
 \end{align}
\begin{align}
 h_{1j}(x)&=p(x)-D\mu_jq_2(x,1), \label{eq:h1}\\
 h_{2j}(x)&=D\mu_jq_2(x,0)-D\lambda_j\phi_jq_1(x,0), \label{eq:h2}\\
 h_{3j}(x,\xi)&=D(\lambda_0-\lambda_j)\partial_\xi q_1(x,\xi)+D(\sigma^-_0-\sigma^-_j)q_2(x,\xi), \label{eq:h3}\\
 h_{4j}(x,\xi)&=D(\mu_0-\mu_j)\partial_\xi q_2(x,\xi)+D(\sigma^+_0-\sigma^+_j)q_1(x,\xi). \label{eq:h4}
\end{align}
All these perturbation coefficients vanish when $\delta_j=\delta_0$.

 {
\begin{lemma}\label{boundf}
    Assume that the kernel bounds in Proposition~\ref{thm:kernelbound} hold, Then there exists a constant \(M_0>0\), depending only on the admissible parameter bounds, the nominal vector \(\delta_0\), the delay \(D\), and the kernel bounds \(\bar K,\bar K_1,\bar q,\bar q_\xi\), such that for every joint mode \(j\in\Rcal\), we have the following bound
    \begin{align}
    &|f_{1j}|+|f_{2j}(x)|+|f_{3j}(x,\xi)|+|f_{4j}(x,\xi)|
    \le M_0\Delta_j, \label{eq:f_bound_group}\\
    &|g_{1j}|+|g_{2j}(x)|+|g_{3j}(x,\xi)|+|g_{4j}(x,\xi)|
    \le M_0\Delta_j, \label{eq:g_bound_group}\\
    &|h_{1j}(x)|+|h_{2j}(x)|+|h_{3j}(x,\xi)|+|h_{4j}(x,\xi)|
    \le M_0\Delta_j. \label{eq:h_bound_group}
    \end{align}
    Consequently, each individual coefficient in \eqref{eq:f1}--\eqref{eq:h4} is bounded by \(M_0\Delta_j\), after increasing \(M_0\) if necessary.
\end{lemma}
}
\begin{proof}
    For the bound of the functions $f_{ij}$, $g_{ij}$, $i \in \{1,2,3,4\}$, we could find the bound of them in previous results~\cite{ZhangOperatorLearning2026,auriol2023mean}. The backstepping kernels $K^{\cdot \cdot}$ are well-defined and bounded such that the bound for $f_{ij}$ and $g_{ij}$ is obtained. Using the same method, for the function $h_{1j}(x)$, we have $h_1 = p(x) - D\mu_jq_2(x,1) = D(\mu_0-\mu_j)q_2(x,1)\leq  D \bar{q}\Delta_j$. 
    We next estimate the coefficients $h_{2j}$, $h_{3j}$, and $h_{4j}$. By the nominal boundary condition of the  {delay-compensating} kernel, $\frac{\mu_0}{\phi_0\lambda_0}q_2(x,0)=q_1(x,0)$, we have $\mu_0 q_2(x,0)-\lambda_0\phi_0 q_1(x,0)=0$. Therefore,
    \begin{align}
    h_{2j}(x)=D(\mu_j-\mu_0)q_2(x,0)
    +D(\lambda_0\phi_0-\lambda_j\phi_j)q_1(x,0).
    \end{align}
Using the uniform bound $|q_i(x,\xi)|\le \bar q$, $i=1,2$, and $|\lambda_0\phi_0-\lambda_j\phi_j|
\le
\bar\phi|\lambda_j-\lambda_0|
+\bar\lambda|\phi_j-\phi_0|$,
we obtain
\begin{align}
|h_{2j}(x)|
&\le
D\bar q|\mu_j-\mu_0|
+D\bar q|\lambda_0\phi_0-\lambda_j\phi_j| \nonumber\\
&\le
D\bar q\Big(
|\mu_j-\mu_0|
+\bar\phi|\lambda_j-\lambda_0|
+\bar\lambda|\phi_j-\phi_0|
\Big) \nonumber\\
&\le
D\bar q\max\{1,\bar\phi,\bar\lambda\}\Delta_j .
\end{align}
For $h_{3j}$, by the definition of $h_{3j}$ and the bounds
$|\partial_\xi q_1|\le \bar q_\xi$ and $|q_2|\le \bar q$, we have
\begin{align}
|h_{3j}(x,\xi)|
&=
\left|
D(\lambda_0-\lambda_j)\partial_\xi q_1(x,\xi)
+D(\sigma_0^--\sigma_j^-)q_2(x,\xi)
\right| \nonumber\\
&\le
D\bar q_\xi|\lambda_j-\lambda_0|
+D\bar q|\sigma_j^--\sigma_0^-| \nonumber\\
&\le
D\max\{\bar q_\xi,\bar q\}\Delta_j .
\end{align}
Similarly, for $h_{4j}$, using $|\partial_\xi q_2|\le \bar q_\xi$
and $|q_1|\le \bar q$, we obtain
\begin{align}
|h_{4j}(x,\xi)|
&=
\left|
D(\mu_0-\mu_j)\partial_\xi q_2(x,\xi)
+D(\sigma_0^+-\sigma_j^+)q_1(x,\xi)
\right| \nonumber\\
&\le
D\bar q_\xi|\mu_j-\mu_0|
+D\bar q|\sigma_j^+-\sigma_0^+| \nonumber\\
&\le
D\max\{\bar q_\xi,\bar q\}\Delta_j .
\end{align}
Consequently, defining $M_h:=
D\max\{
\bar q\max\{1,\bar\phi,\bar\lambda\},
\bar q_\xi,
\bar q\}$, we have, for all $j\in\Rcal$ and all admissible $(x,\xi)$, $|h_{2j}(x)|\le M_h\Delta_j$, $|h_{3j}(x,\xi)|\le M_h\Delta_j$, and $|h_{4j}(x,\xi)|\le M_h\Delta_j$. Taking $M_0$ as the maximum of these bounds for $f_i,g_i,h_i$, we obtain the desired result. This completes the proof of Lemma~\ref{boundf}.
\end{proof}
The nominal transformations \eqref{eq:plant_BS} and \eqref{eq:delay_BS} are boundedly invertible on $L^2([0,1])^3$. Hence there exist constants $0<c_T\le C_T$ such that
\begin{align}\label{eq:norm_equiv_transform}
    c_T\norm{(u,v,\eta)}_{L^2}^2
\le \norm{(\alpha,\beta,z)}_{L^2}^2
\le C_T\norm{(u,v,\eta)}_{L^2}^2. 
\end{align}

\section{Lyapunov Analysis with Mode-independent Lyapunov function}
\label{sec:LyapunovAnalysis}
To analyze the stability of the closed-loop system under the nominal  {delay-compensating} controller, we construct a mode-independent Lyapunov functional that captures the energy of the system state and quantifies the effect of parameter mismatch. The Lyapunov functional is defined as
\begin{align}
    V(Y)=\int_0^1\e^{-\rho x}\alpha^2\,\dd x
 +a\int_0^1\e^{\rho x}\beta^2\,\dd x
 +bD\int_0^1\e^{\rho_zDx}z^2\,\dd x, \label{eq:commonV}
\end{align}
where $\rho,\rho_z,a,b>0$. It is equivalent to the $L^2$ norm: there exist $0<c_V\le C_V$ such that
\begin{align}\label{eq:V_equiv}
    c_V\norm{Y}_{L^2}^2\le V(Y)\le C_V\norm{Y}_{L^2}^2. 
\end{align}
We then consider the infinitesimal generator $\cL$ of the Lyapunov candidate $V$ defined in~\eqref{eq:commonV} as~\cite{ross2014introduction}
\begin{align}
 \cL V(Y,\delta) =&\limsup _{\Delta t \rightarrow 0^{+}} \frac{1}{\Delta t} ( \mathbb{E}(V(Y(t+\Delta t)|\delta(t) = \delta_j))\nonumber\\
&-V(Y(t))).
\end{align}
The infinitesimal generator of the joint process $(Y(t),\delta(t))$ acting on the common functional at each Markov mode $j \in \{1,\dots, r\}$ where $\delta(t) = \delta_j$, can be computed as
\begin{align}
 \cL V(Y,\delta_j)=\mathcal D V(Y)[A_jY+P_jY], \label{eq:generator_common}
\end{align}
where $A_jY$ denotes the pure transport part of \eqref{eq:tarsto1}-\eqref{eq:tar_bc3}, $P_jY=(F_j,G_j,H_j/D)$ denotes the perturbation contribution, and $\mathcal D V(Y)$ is the Fr\'echet derivative with respect to the infinite-dimensional state. There is no term of the form $\sum_\ell\tau_{j\ell}(V_\ell-V_j)$ because the same $V$ is used in all modes and the state is not reset at Markov jumps. This is the main structural difference from a mode-dependent Lyapunov proof, the price paid here is that the coefficients in the boundary inequalities below must be selected uniformly over all modes. Then we have the following lemma that provides a sufficient condition for the mean-square exponential stability of the closed-loop system under the nominal  {delay-compensating} controller.
\begin{lemma}\label{lem:generator}
    Let $c_1,c_2,c_3,c_4$ be positive parameters of Young's inequality.  {Choose \(\rho,\rho_z,a,b>0\) such that}
\begin{align}
    d_\beta&:=\e^\rho,
    d_z:=\e^{\rho_zD},\\
    Q_0&:=\norm{q_1(0,\cdot)}_{L^2}^2
    +\norm{q_2(0,\cdot)}_{L^2}^2,\\
    \gamma_\beta&:=a\underline\mu-\bar\lambda\bar\phi^2 > 0,\label{eq:gammabeta}\\
    \gamma_z&:=b-a\bar\mu\e^\rho >0\label{eq:gammaz}.
\end{align}
Define
\begin{align}
M_V:=&(\frac{3+3ad_{\beta} + 2bd_z}{c_V})(1+\frac{1}{c_T}) + \frac{1}{ c_1 c_V} + \frac{ad_{\beta}}{c_V} + \frac{bd_z}{c_V c_3}\nonumber\\
&+ \frac{3bd_zc_3 Q_0}{c_Tc_V} +\frac{bd_z}{c_Vc_4}, \label{eq:MV_common}\\
M_\beta:=&c_1 + {ad_{\beta}c_2}+bd_zc_4, \label{eq:Mbeta_common}\\
M_z:=&3bd_zc_3 . \label{eq:Mz_common}
\end{align}
 {Then, for every joint mode \(j\in\Rcal\), the generator satisfies}
\begin{align}\label{eq:common_generator_combined}
\cL V(Y,\delta_j)
\le&-\bigl(\nu-M_0 M_V\Delta_j\bigr) V(Y) \nonumber\\
&-\bigl(\gamma_\beta-M_0 M_\beta\Delta_j\bigr)\beta^2(0,t) \nonumber\\
&-\bigl(\gamma_z-M_0 M_z\Delta_j\bigr)z^2(0,t),
\end{align}
where
\begin{align}\label{eq:eta0_common}
    \nu:=\min\{\rho\underline\lambda,\rho\underline\mu,\rho_z\}>0 .
\end{align}
Consequently, if
\begin{align}\label{eq:uniform_condition}
\Delta_{\max}<\eps^\star
:=\min\left\{
\frac{\nu}{2M_0 M_V},
\frac{\gamma_\beta}{2M_0 M_\beta},
\frac{\gamma_z}{2M_0 M_z}
\right\},
\end{align}
then for each fixed mode \(j\in\Rcal\), along classical solutions of the target system \eqref{eq:tarsto1}-\eqref{eq:tar_bc3}, the Lyapunov functional satisfies $\tfrac{d}{dt}V(Y(t))\le -\kappa V(Y(t))$.
Equivalently, since \(V\) is common to all modes and the state is not
reset at jumps, this inequality coincides with the generator inequality
\begin{align}\label{eq:common_gen_est}
 \cL V(Y,\delta_j)\le-\kappa V(Y),
 \qquad
 \kappa:=\nu-M_0 M_V\Delta_{\max}>0 .
\end{align}
\end{lemma}

\begin{proof}
    The Fr\'echet derivative of \eqref{eq:commonV} in direction $\Psi=(\psi_\alpha,\psi_\beta,\psi_z)$ is
\begin{align}
 \mathcal D V(Y)[\Psi]
 &=2\int_0^1\e^{-\rho x}\alpha\psi_\alpha\,\dd x
 +2a\int_0^1\e^{\rho x}\beta\psi_\beta\,\dd x \nonumber\\
 &+2bD\int_0^1\e^{\rho_zDx}z\psi_z\,\dd x. \label{eq:Frechet}
\end{align}
Substituting first only the transport terms in \eqref{eq:tarsto1}-\eqref{eq:tar_bc3} and integrating by parts yields
\begin{align}
 \mathcal D V(Y)[A_jY]
 &=-\rho\lambda_j\int_0^1 \e^{-\rho x}\alpha^2\,\dd x
   -a\rho\mu_j\int_0^1\e^{\rho x}\beta^2\,\dd x \nonumber\\
 &-\rho_zbD\int_0^1\e^{\rho_zDx}z^2\,\dd x
   -\lambda_j\e^{-\rho}\alpha^2(1,t) \nonumber\\
 &+(\lambda_j\phi_j^2-a\mu_j)\beta^2(0,t)
   +(a\mu_j\e^\rho-b)z^2(0,t). \label{eq:pure_derivative}
\end{align}
Then we have 
\begin{align}
 \mathcal D V(Y)[A_jY]
 &\le-\nu V(Y) + (\lambda_j\phi_j^2-a\mu_j)\beta^2(0,t) \nonumber\\
 &+ (a\mu_j\e^\rho-b) z^2(0,t)-\lambda_j\e^{-\rho}\alpha^2(1,t), \label{eq:pure_bound}
\end{align}
where $\nu:=\min\{\rho\underline\lambda,\rho\underline\mu,\rho_z\}$.

Then substituting the perturbation terms and using similar estimates, we have
    \begin{align}
        \mathcal D V(Y)[P_jY]=&2\int_0^1{\e^{-\rho x}}\alpha(x)\mathcal F_j(x)\,\dd x\nonumber\\
        &+2a\int_0^1 \e^{\rho x}\beta(x)\mathcal G_j(x)\,\dd x \nonumber\\
        &+2 b \int_0^1 \e^{\rho_z Dx}z(x)\mathcal H_j(x)\,\dd x. \label{eq:Ij}
    \end{align}
    Next, we will bound the term $\mathcal D V(Y)[P_jY]$ using the bounds on $f_{ij}$, $g_{ij}$, and $h_{ij}$ from Lemma~\ref{boundf}. For the first term in $\mathcal D V(Y)[P_jY]$, using Young's inequality, we have \(2\int_0^1 \abs{\e^{-\rho x}\alpha(x)f_{1j} v(x,t)\,\dd x} \leq \frac{M_0\Delta_j}{c_V}\left(1+ \frac{1}{c_T}\right)V(Y)\). Then for the second term in $\mathcal{F}_j(x)$, we have \(2\int_0^1 \abs{\e^{-\rho x}\alpha(x)f_{2j}(x)\beta(0,t)\,\dd x}\leq \frac{M_0\Delta_j}{c_V c_1}V(Y) + c_1M_0\Delta_j\beta^2(0,t)\). For the third term in $\mathcal{F}_j(x)$, we have \(2\int_0^1 \abs{\e^{-\rho x}\alpha(x)\int_0^x f_{3j}(x,\xi)u(\xi,t)\,d\xi\,\dd x}\leq \frac{M_0\Delta_j}{c_V}\left(1+\frac{1}{c_T}\right)V(Y)\). For the fourth term in $\mathcal{F}_j(x)$, we have \(2\int_0^1 \abs{\e^{-\rho x}\alpha(x)\int_0^x f_{4j}(x,\xi)v(\xi,t)\,d\xi\,\dd x}\leq \frac{M_0\Delta_j}{c_V}\left(1+\frac{1}{c_T}\right)V(Y)\). 
    
    For the terms in $\mathcal{G}_j(x)$, using the same method, we have \(2a\int_0^1 \abs{\e^{\rho x }\beta(x)g_{1j}u(x,t)\,\dd x}\leq \frac{a d_{\beta} M_0\Delta_j}{c_V} (1+\frac{1}{c_T})V(Y)\),
    where $d_{\beta}$ is the upper bound of the term $\e^{\rho x}, x\in[0,1]$.
    For the second term in $\mathcal{G}_j(x)$, we have \(2a\int_0^1 \abs{\e^{\rho x}\beta(x)g_{2j}(x)\beta(0,t)\,\dd x}\leq \frac{ad_{\beta} M_0\Delta_j}{c_V c_2}V(Y) + {ad_{\beta} c_2M_0\Delta_j}\beta^2(0,t)\). For the third term in $\mathcal{G}_j(x)$, we have \(2a\int_0^1 \abs{\e^{\rho x}\beta(x)\int_0^x g_{3j}(x,\xi)u(\xi,t)\,d\xi\,\dd x}\leq \frac{ad_{\beta}M_0\Delta_j}{c_V} (1+\frac{1}{c_T})V(Y)\). For the fourth term in $\mathcal{G}_j(x)$, we have \(2a\int_0^1 \abs{\e^{\rho x}\beta(x)\int_0^x g_{4j}(x,\xi)v(\xi,t)\,d\xi\,\dd x}\leq \frac{ad_{\beta}M_0\Delta_j}{c_V} (1+\frac{1}{c_T})V(Y)\). 
    
    For the first term in $\mathcal{H}_j(s)$, we obtain \(2b\int_0^1 \abs{\e^{\rho Dx}z(x)h_{1j}(x)\eta(0,t)\,\dd x}\leq \frac{bd_z M_0\Delta_j}{c_V c_3} V(Y) + bd_z c_3M_0\Delta_j \eta^2(0,t)\),
    where $d_z$ is an upper bound of the term $\e^{\rho Dx},x\in[0,1]$. Using the nominal backstepping transformation for the delay state, evaluating at $x=0$, we have $\eta(0,t) = z(0,t) + \int_0^1 q_1(0,\xi)u(\xi,t) + q_2(0,\xi)v(\xi,t)\,\dd\xi$. Next, we can bound the term $\eta^2(0,t)$ by $3z^2(0,t) + 3\left(\int_0^1 q_1(0,\xi)u(\xi,t)\,\dd\xi\right)^2 + 3\left(\int_0^1 q_2(0,\xi)v(\xi,t)\,\dd\xi\right)^2$. Thus the term involving $\eta^2(0,t)$ can be bounded by $bd_z c_3 M_0\Delta_j\eta^2(0,t)\leq 3bd_z c_3 M_0\Delta_j z^2(0,t)+ \frac{3bd_z c_3 Q_0}{c_T c_V}M_0\Delta_j V(Y)$.
    Finally, we obtain the bound for the first term in $\mathcal{H}_j(s)$ as \(2b\int_0^1 \abs{\e^{\rho Dx}z(x)h_{1j}(x)\eta(0,t)\,\dd x}\leq \left(\frac{bd_z}{c_V c_3}+\frac{3bd_z c_3 Q_0}{c_T c_V}\right) M_0\Delta_j V(Y)+ 3bd_z c_3 M_0\Delta_j z^2(0,t)\). For the second term in $\mathcal{H}_j(s)$, we have \(2b\int_0^1 \abs{\e^{\rho Dx}z(x)h_{2j}(x)\beta(0,t)\,\dd x}\leq \frac{bd_z}{c_V c_4}M_0\Delta_j V(Y) + bd_z c_4 M_0\Delta_j\beta^2(0,t)\),
    where $c_4$ is an arbitrary positive constant. 
    For the last two terms in $\mathcal{H}_j(s)$, we obtain the bound for $h_{3j}(x,\xi)$ as \(2b\int_0^1 \abs{\e^{\rho Dx}z(x)\int_0^1 h_{3j}(x,\xi)u(\xi,t)\,d\xi\,\dd x}\leq \frac{bd_z}{c_V}M_0\Delta_j(1+\frac{1}{c_T})V(Y)\). For the fourth term in $\mathcal{H}_j(s)$, we have \(2b\int_0^1 \abs{\e^{\rho Dx}z(x)\int_0^1 h_{4j}(x,\xi)v(\xi,t)\,d\xi\,\dd x}\leq \frac{bd_z}{c_V}M_0\Delta_j(1+\frac{1}{c_T})V(Y)\).
    Therefore, we have the following result for the generator,
    \begin{align}
        &\mathcal{D} V_j(Y)\big[\mathcal A_jY + \mathcal{P}_j Y\big] \nonumber\\
        &\leq -(\nu - M_0 M_V\Delta_j) V(Y) -\lambda_j\e^{-\rho}\alpha^2(1,t) \nonumber\\
        &-(a\mu_j- \lambda_j\phi_j^2 - M_0 M_\beta\Delta_j)\beta^2(0,t)\nonumber\\
        &-\big(b - a\mu_je^\rho - M_0 M_z\Delta_j\big)z^2(0,t),
    \end{align}
     {Since \(\lambda_j\le\bar\lambda\), \(|\phi_j|\le\bar\phi\), and \(\mu_j\ge\underline\mu\), we have
    }
     {\begin{align}
        a\mu_j-\lambda_j\phi_j^2-M_0M_\beta\Delta_j
        &\ge \gamma_\beta-M_0M_\beta\Delta_j,\\
        b-a\mu_j\e^\rho-M_0M_z\Delta_j
        &\ge \gamma_z-M_0M_z\Delta_j.
    \end{align}
    }
     {If \eqref{eq:uniform_condition} holds, then
    \begin{align}
        \nu-M_0M_V\Delta_j&\ge \nu-M_0M_V\Delta_{\max}=\kappa>0,\\
        \gamma_\beta-M_0M_\beta\Delta_j&\ge \gamma_\beta/2>0,\\
        \gamma_z-M_0M_z\Delta_j&\ge \gamma_z/2>0.
    \end{align}
    }
    This completes the proof of Lemma~\ref{lem:generator}.
\end{proof}
 {
\begin{remark}
The parameter condition in \eqref{eq:gammabeta}, \eqref{eq:gammaz} are feasible and practical to check. Indeed, one may first choose \(\rho>0\) and \(a>\bar\lambda\bar\phi^2/\underline\mu\), which gives \(\gamma_\beta=a\underline\mu-\bar\lambda\bar\phi^2>0\). Then \(b>a\bar\mu\e^\rho\) can be selected so that \(\gamma_z=b-a\bar\mu\e^\rho>0\). 
However, the small-mismatch condition in \eqref{eq:uniform_condition} is more conservative and less direct to verify because it depends on the constants \(M_V\), \(M_\beta\), and \(M_z\) that are generated by the Young-inequality estimates and the norm-equivalence constants. In practice, one may directly select \(\gamma_\beta\) and \(\gamma_z\) to be sufficiently large by tuning the parameters \(\rho\), \(a\), and \(b\) as described above, which gives a more direct design step for ensuring a positive decay rate \(\kappa\) in \eqref{eq:common_gen_est}. Then the small-mismatch condition in \eqref{eq:uniform_condition} can be used as a theoretical margin for the admissible mismatch.
\end{remark}
}

\begin{theorem}\label{thm:main}
    Consider the stochastic closed-loop system~\eqref{eq:stosys1}-\eqref{eq:stosys4} under the nominal  {delay-compensating} controller \eqref{eq:nomcontrol}. Suppose that the conditions of Lemma~\ref{lem:generator} hold and that the admissible Markov modes satisfy the uniform small-mismatch condition~\eqref{eq:uniform_condition}. Then the closed-loop system is pathwise exponentially stable such that for every initial condition $W^0:=(u^0,v^0,\eta^0)\in L^2([0,1])^3$, and every initial Markov mode \(\theta(0)=j^0\in\mathcal R\), the corresponding solution satisfies 
    \begin{align}
        V(Y(t,\omega))\le e^{-\kappa t}V(Y(0)), t\ge 0,
    \end{align}
    for almost every sample path \(\omega\). Consequently,
    \begin{align}
        \|W(\cdot,t,\omega)\|_{L^2}^2
    \le M e^{-\kappa t} \|W^0\|_{L^2}^2,
    \end{align}
    for almost every sample path \(\omega\), where $M:=\frac{C_VC_T}{c_Vc_T}$. In particular, the closed-loop system is mean-square exponentially stable
    \begin{align}
        \mathbb E_{(W^0,j^0)}
    \left[
    \|W(\cdot, t)\|_{L^2}^2
    \right]
    \le
    M e^{-\kappa t}
    \|W^0\|_{L^2}^2 .
    \end{align}
\end{theorem}
\begin{proof}
    We first prove the result for classical solutions. Fix a sample path
\(\omega\) outside the null set on which the Markov chain is explosive.
Let $0=T_0<T_1<T_2<\cdots$ denote the jump times of \(\theta(t,\omega)\). By the non-explosion of the finite-state Markov process, only finite jumps occur on every time interval. On each interval \([T_k,T_{k+1})\), the active
mode is constant. Denote this mode by \(j_k\), namely $\theta(t,\omega)=j_k, t\in [T_k,T_{k+1})$. Along this inter-jump interval, the stochastic target system reduces to
the deterministic target system associated with mode \(j_k\). By Lemma~\ref{lem:generator}
and the small-mismatch condition~\eqref{eq:uniform_condition}, we have
\begin{align}
    \frac{d}{dt}V(Y(t,\omega)) \le -\bigl(\nu-M_0M_V\Delta_{j_k}\bigr)V(Y(t,\omega))
\end{align}
up to non-positive boundary terms. Since $\Delta_{j_k}\le \Delta_{\max}$, it follows that
\begin{align}
    \frac{d}{dt}V(Y(t,\omega)) \le -\kappa V(Y(t,\omega)), t\in [T_k,T_{k+1}).
\end{align}
Hence, by Gronwall's inequality,
\begin{align}
    V(Y(t,\omega)) \le \e^{-\kappa(t-T_k)}V(Y(T_k,\omega)), t\in [T_k,T_{k+1}).
\end{align}

At a jump time \(T_k\), the Markov mode changes but the PDE state is not reset. Moreover, the backstepping transformation used here is the nominal one and is independent of the Markov mode. Therefore, $Y(T_k,\omega)=Y(T_k^-,\omega)$, and consequently $V(Y(T_k,\omega))=V(Y(T_k^-,\omega))$. Concatenating the above estimate over all inter-jump intervals gives
\begin{align}
    V(Y(t,\omega))\le \e^{-\kappa t}V(Y(0)), t\ge 0,
\end{align}
for almost every sample path \(\omega\). Using the norm equivalences~\eqref{eq:norm_equiv_transform} and~\eqref{eq:V_equiv},
we obtain
\begin{align}
    \|W(\cdot,t,\omega)\|_{L^2}^2 
    \le \frac{1}{c_Vc_T}e^{-\kappa t}V(Y(0)) \le M \e^{-\kappa t} \|W^0\|_{L^2}^2 
\end{align}
This proves the pathwise estimate with $M=\tfrac{C_VC_T}{c_Vc_T}$. Since the above estimate holds for almost every sample path and the right-hand side is deterministic for fixed \(W^0\) and \(j^0\), taking conditional expectations gives
\begin{align}
    \mathbb E_{(W^0,j^0)}
    \left[\|W(\cdot,t)\|_{L^2}^2 \right] \le M \e^{-\kappa t}
    \|W^0\|_{L^2}^2 .
\end{align}
This finishes the proof of Theorem~\ref{thm:main}.
\end{proof}

\begin{remark}
    Compared to~\cite{auriol2023mean,ZhangOperatorLearning2026}, the main contribution of Theorem~\ref{thm:main} is to show that the nominal  {delay-compensating} controller stabilizes the original stochastic system with Markov-jumping parameters and input delay. The key technical novelty is to construct a common Lyapunov functional that is independent of the Markov mode, which allows us to avoid the jump term in the generator and obtain a more direct stability estimate. The price paid for this approach is that the coefficients in the boundary inequalities must be selected uniformly over all modes, which leads to a more conservative stability condition compared to a mode-dependent Lyapunov functional. However, this approach also provides a more unified and simpler analysis framework for systems with Markov-jumping parameters.
\end{remark}

\section{Simulation}
\label{sec:simulation}
In the simulation, we use the joint Markov chain formulation in Section~\ref{sec:problemstatement} with \(r=3\). The Markov-jumping parameter vector \(\delta(t)\) takes values in \(\{\delta_1,\delta_2,\delta_3\}\), where $\delta_1 = (\lambda_1=0.85, \mu_1=0.85, \sigma^-_1=0.85, \sigma^+_1=0.85,\phi_1=0.85)$, $\delta_2 = (\lambda_2=1, \mu_2=1, \sigma^-_2=1, \sigma^+_2=1,\phi_2=1)$, $\delta_3 = (\lambda_3=1.15, \mu_3=1.15, \sigma^-_3=1.15, \sigma^+_3=1.15,\phi_3=1.15)$
The second mode \(\delta_2\) is selected as the nominal vector
\(\delta_0\), and the input delay is \(D=3\,{\rm s}\).  {Hence \(\Delta_1=\Delta_3=5\times0.15=0.75\), \(\Delta_2=0\), and \(\Delta_{\max}=0.75\).} The initial probabilities are set as $(0.32,0.36,0.32)$, the transition rates $\tau_{ij}$ are defined as
\begin{align}
    \tau_{i j}(t)= \begin{cases}0, \quad \text { if } i=j \\ 
        20, \quad \text { if } i=1,j\in\{2,3\} \\ 
        10, \quad  \text { if } i\in\{2,3\},j =1\\ 
        10+20 \cos (0.01(i+3j) t)^2, \text{others}  \end{cases}
\end{align}
Using the above settings, we numerically solve the Kolmogorov equation and obtain the results of the probability of each mode in the whole time period, as shown in Figure~\ref{fig:proandmode}(a). We also simulate a sample path of the Markov mode $\delta(t)$, which is shown in Figure~\ref{fig:proandmode}(b). 
\begin{figure}[htbp]
    \centering
    \subfloat[Probability]{\includegraphics[width=0.45\linewidth]{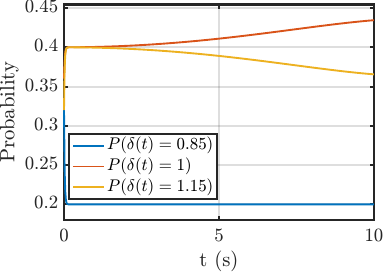}}
    \subfloat[Mode path]{\includegraphics[width=0.45\linewidth]{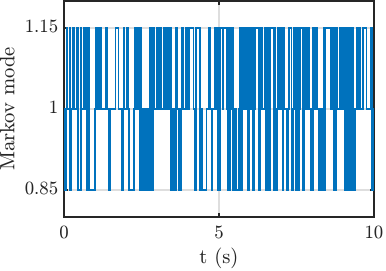}}
    \caption{The Kolmogorov forward equation probabilities and a sample Markov mode path}
    \label{fig:proandmode}
\end{figure}
\begin{figure}[htbp]
    \centering
    \subfloat[Norm comparison]{\includegraphics[width=0.46\linewidth]{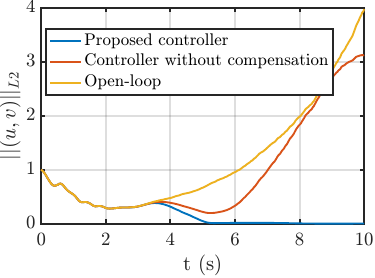}}
    \subfloat[Control input]{\includegraphics[width=0.47\linewidth]{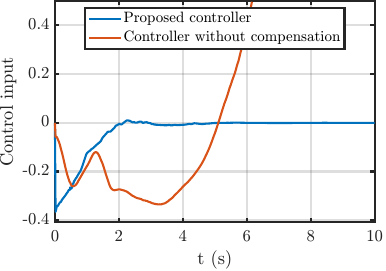}}
    \caption{Comparison of state norm and control input}
    \label{fig:comparison-normcontrol}
\end{figure}
\begin{figure}[htbp]
    \centering
    \subfloat[State norm]{\includegraphics[width=0.48\linewidth]{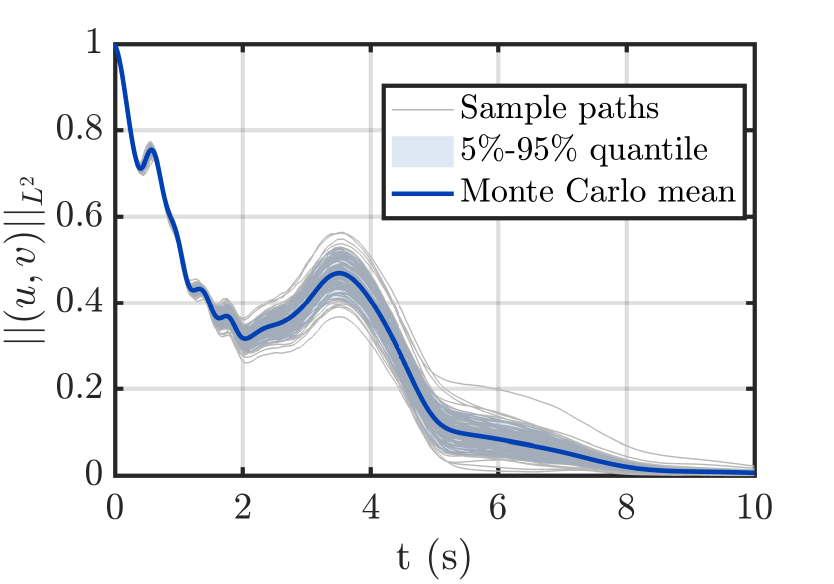}}
    \subfloat[Total norm with exponential decay]{\includegraphics[width=0.48\linewidth]{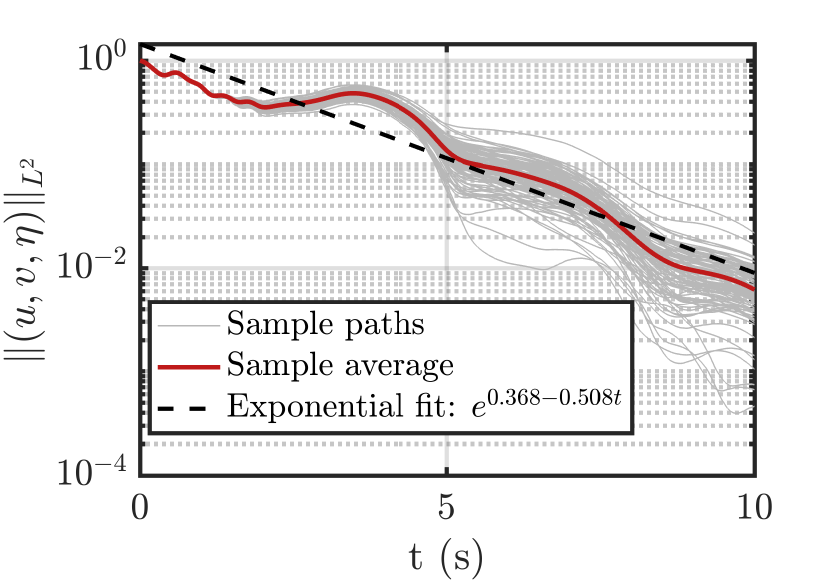}}
    \caption{State norm and the exponential decay for 100 Monte Carlo runs}
    \label{fig:monte-carlo-norm}
\end{figure}
The initial condition of the system is set to $u^0(x) = \sin(2\pi x)$, $v^0(x) = \cos(2\pi x)$,  {and \(\eta^0(x)=0\)}.  {The comparison controller without delay compensation is the nominal backstepping boundary controller $U_{\rm nd}(t)=\int_0^1K^{vu}(1,\xi)u(\xi,t)\,\dd\xi
+\int_0^1K^{vv}(1,\xi)v(\xi,t)\,\dd\xi$, applied directly to the delayed plant}.
 {For the Lyapunov parameters, one admissible boundary-weight selection is \(\rho=\rho_z=0.1\), \(a=2\), and \(b=3\). With \(\underline\mu=0.85\), \(\bar\lambda=\bar\mu=\bar\phi=1.15\), this gives $\gamma_\beta=2\times(0.85)-1.15\times(1.15)^2\approx0.179>0$, $\gamma_z=3-2(1.15)e^{0.1}\approx0.458>0$.
These values verify the positivity of the nominal boundary dissipation terms. The full sufficient condition \eqref{eq:uniform_condition}, however, further requires computable upper bounds for \(M_0\), \(M_V\), \(M_\beta\), \(M_z\), \(c_T\), and \(c_V\). Because these constants depend on kernels and inverse-transformation bounds, the present numerical example is reported as an empirical illustration unless the right-hand side \(\eps^\star\) in \eqref{eq:uniform_condition} is evaluated from the computed kernels.}
Figure~\ref{fig:comparison-normcontrol}(a) shows the evolution of the $L^2$ norm of the state $(u(\cdot,t), v(\cdot,t))$ over time for one sample path.  {The nominal delay-compensating controller stabilizes the system, whereas the open-loop response and the controller without compensation exhibit substantially growth}. Figure~\ref{fig:comparison-normcontrol}(b) shows the comparison of the controller without compensation and the delay-compensating controller over time. Figure~\ref{fig:monte-carlo-norm}(a) shows the $L^2$ norm of the state $(u(\cdot,t), v(\cdot,t))$ for 100 Monte Carlo runs, which further confirms the stability of the closed-loop system under the nominal  {delay-compensating} controller across different sample paths. Figure~\ref{fig:monte-carlo-norm}(b) shows the mean value of the total $L^2$ norm of the state across the 100 Monte Carlo runs in semi-logarithmic scale.  {The fitted empirical decay rate is \(0.508\). This empirical rate is included only to summarize the numerical trend, it is not claimed to coincide with the conservative theoretical rate \(\kappa\) in Theorem~\ref{thm:main}.}

\section{Conclusion}
\label{sec:conclusion}
This paper has invesgated the robust stabilization of \(2\times2\) linear Markov-jumping hyperbolic PDEs with boundary input delay. A nominal  {delay-compensating} backstepping controller was designed and applied to the Markov-jumping system. A common mode-independent Lyapunov functional was constructed to derive a uniform small-mismatch condition. The condition guarantees the mean-square exponential stability of the closed-loop system. Future work will focus on extending the design to output-feedback and learning-based control schemes.



\bibliographystyle{abbrv}
\bibliography{reference}

\end{document}